\documentclass{llncs}
\usepackage{underscore}           	
\usepackage{amsmath,amsfonts,amssymb,times,graphicx,color}
\usepackage{mathtools}  
\usepackage{url}
\usepackage{hyperref} 
\usepackage{tikz}
\usepackage{subcaption} 
\usepackage{verbatim}
\usepackage{ifthen}
\usetikzlibrary{shapes.geometric}
\mathtoolsset{showonlyrefs}

\newcommand{\couic}[1]{}
\newcommand{\couicfootnote}[1]{}
\newcommand{\couicefootnote}[1]{}

\newcommand{\eg}{e.g.}

\newcommand{\ie}{i.e.}

\def\gaugecol{red}
\newcommand{\gfhei}{0.4}
\newcommand{\gfwid}{0.6}
\newcommand{\ovalTikz}[6]{ 
\begin{scope}
  \ifthenelse{#3 > 0 \AND #4>0}{\def\mycolc{white}}{\def\mycolc{\gaugecol}};
  \draw[color=\mycolc, thick] (#5 ,#6) ellipse (#1 and #2);
  \clip (#5 ,#6) ellipse (#1 and #2);
  
  \ifthenelse{#3 > 0}{\def\mycola{\gaugecol}}{\def\mycola{white}};
  \fill[color=\mycola] (#5 - #1 ,#6 - #2) rectangle (#5 ,#6 + #2);
  \ifthenelse{#4 > 0}{\def\mycolb{\gaugecol}}{\def\mycolb{white}};
  \fill[color=\mycolb] (#5 + #1 ,#6 - #2) rectangle (#5 ,#6 + #2);
  
  \draw[color=\mycolc, thick] (#5 ,#6 - #2) -- (#5 ,#6 + #2);
\end{scope}
}
\newcommand{\stateTikz}[6]{
  \ifthenelse{#5 > 0 \AND #6>0}{\def\mycola{white}}{\def\mycola{black}};
  \ifthenelse{#5 > 0}{\def\mycolb{black}}{\def\mycolb{white}};
  \ifthenelse{#6 > 0}{\def\mycolc{black}}{\def\mycolc{white}};
  \filldraw[color=\mycola, fill=\mycolb, thick](#1- #3 , #2) rectangle (#1 , #2 + #4);
  \filldraw[color=\mycola, fill=\mycolc, thick](#1, #2) rectangle (#1 + #3 , #2 + #4);
}

\title{A gauge-invariant reversible cellular automaton}

\author{Pablo Arrighi\inst{1,2} \and Giuseppe Di Molfetta\inst{1,3} \and Nathana\"el Eon\inst{1,4}}

\tocauthor{Pablo Arrighi, Giuseppe Di Molfetta, Nathana\"el EON}

\institute{Aix-Marseille Univ, Universit\'e de Toulon, CNRS, LIS, Marseille, France
\and
IXXI, Lyon, France
\and
Departamento de F{\'{i}}sica Te{ó}rica and IFIC, Universidad de Valencia-CSIC, Dr. Moliner 50, 46100-Burjassot, Spain
\and 
\'Ecole Centrale, France
}

\titlerunning{A gauge-invariant RCA}
\authorrunning{P. Arrighi, G. Di Molfetta, N. Eon}

\begin{document}
\pagestyle{headings}  
\tikzstyle{losange} = [diamond, draw, text badly centered, inner sep=0.2cm, aspect=2]
\maketitle

\begin{abstract}
Gauge-invariance is a fundamental concept in physics---known to provide mathematical justifications for the fundamental forces. In this paper, we provide discrete counterparts to the main gauge theoretical concepts, directly in terms of Cellular Automata. More precisely, we describe a step-by-step gauging procedure to enforce local symmetries upon a given Cellular Automaton. We apply it to a simple Reversible Cellular Automaton for concreteness. From a Computer Science perspective, discretized gauge theories may be of use in numerical analysis, quantum simulation, fault-tolerant (quantum) computation. From a mathematical perspective, discreteness provides a simple yet rigorous route straight to the core concepts.
\end{abstract}

\section{Introduction}

In Physics, symmetries act as guiding principles towards discovering the laws we put forward to model nature. Among them, Gauge symmetries are absolutely central, as they provide mathematical justifications for all four fundamental forces: electromagnetism and gravity (long range interactions), weak and strong forces (short range interactions) \cite{quigg2013gauge}. In this paper we express the key notions of gauge theories natively in Computer Science friendly, Discrete Mathematics terms---we do so in order to make them available to these disciplines, and in order to clarify its concepts. More precisely, we describe a discrete counterpart to the gauging procedure. I.e. we provide a step-by-step procedure to enforce local symmetries within Cellular Automata. 

These methods may lead to natural, physics-inspired CA. More importantly, the fields of numerical analysis, quantum simulation, digital physics are constantly looking for discrete schemes that simulate known physics \cite{georgescu2014quantum}. Quite often, these discrete schemes seek to retain the symmetries of the simulated physics; whether in order to justify the discrete scheme as legitimate, or in order to do the Monte Carlo-counting right \cite{hastings1970monte}. Generally speaking, since gauge symmetries are essential in physics, having a discrete counterpart of it may also be. 

Interestingly, this way of enforcing local redundancies also bears some resemblances with error-correction, and echoes the fascinating question of noise resistance within spatially--distributed models of computation \cite{harao1975fault,Toom}, as was pointed out in the context of quantum computation in \cite{kitaev2003fault,nayak2008non}.

Although we authors come from the field of quantum computation and simulation, the formalism we use is totally devoid of any quantum theory, least action principle, nor Lagrangian. The notions here are directly formulated in terms of the discrete dynamical system. We believe that this provides a uniquely direct route to the root concepts. This discrete mathematics framework makes the presentation original, and simpler. But it also allows for more rigorous definitions, that in turn allow us to prove some essential consistency lemmas that are usually left aside. Our running example provides what seems to be the simplest non-trivial Gauge theory so far and illustrates the key concepts. Given the fame of Gauge theories, we think this may be a remarkable pedagogical  asset. 

The paper is organized as follows. In Sec. \ref{sec:gaugeinvariance} we introduce the notions of local transformations which define the desired symmetry, and of gauge-invariance which captures the (non-)compliance of a given Cellular Automaton (CA) with the desired symmetry. In Sec. \ref{sec:gaugefield} we show how a non-gauge-invariant CA can be made gauge-invariant, at the heavy cost of becoming spacetime dependent upon an external parameter, referred to as the gauge field. This new parameter not only implements the symmetry---it leads to new behaviours for the CA. In Sec. \ref{sec:gaugedynamics} the gauge field gets internalized into the configuration space, and a whole family of homogeneous gauge-invariant CA is obtained, leading us to the notions of gauge-fixing and gauge-constraining. A simple Reversible Cellular Automaton (RCA) is used to illustrate each concepts, throughout the paper. In Sec. \ref{sec:conclusion} we summarize, provide related works and perspectives. 

\section{The gauge-invariance requirements}\label{sec:gaugeinvariance}

\paragraph{Theory to be gauged.} In this paper `theories' stand for CA. As our running example, we pick possibly the simplest and most natural physics-like RCA : one that has particles moving left and right. More precisely, each cell of the RCA has a state in $\Sigma= \{ \square\hspace{-0.1em}\square,
\square\hspace{-0.1em}\blacksquare,
\blacksquare\hspace{-0.15em}\square,
\blacksquare\hspace{-0.1em}\blacksquare \} \cong \{ 00 , 01, 10, 11 \}$. 
Its dynamics $R$ is defined through a local rule $\lambda_R$ which computes the next state of a cell from that of its left and right neighbours, i.e. $\psi(x,t+1)=\lambda_R(\psi(x-1,t),\psi(x+1,t))$, with $\psi(x,t)$ the state of cell $x$ at time $t$. A spacetime diagram $\psi:\mathbb{Z}^2\rightarrow\Sigma$, is said to be $R$--valid if and only if it is produced by applying $\lambda_R$, see for instance Fig-\ref{fig:spacetimediagram}. We also say that it is `a solution'. We use the shorthand notations $(R\psi)(x,t)$ for $\psi(x,t+1)$, $\psi(x)$ for $\psi(x,0)$, $\psi(.,t)$ for the function mapping $x$ into $\psi(x,t)$.\\
In our running example, the $R$ that we consider can be expressed in the block circuit form of Fig-\ref{fig:framework}, with $W$ the gate that swaps two bits:
\begin{align*}
\psi(x,t+1)&=\lambda_R(\psi(x-1,t),\psi(x+1,t))\\
&=W(\psi^+(x-1,t)\otimes\psi^-(x+1,t))\\
&=\psi^-(x+1,t)\otimes\psi^+(x-1,t)
\end{align*}
with $\psi(x,t)=(\psi^-(x)\otimes\psi^+(x))$. RCA presented in such a block circuit form are often referred to as (Margolus--)Partitioned CA in Computer Science vocabulary\cite{ToffoliMargolusModelling}, or as Lattice-gas automata in Physics\cite{wolf2004lattice}. This theory is {\em to-be-gauged}. This means that although it may have a global symmetry (here the CA has global black/white--symmetry, see Fig-\ref{fig:bwsym1} $(a)-(b)$), it lacks a certain local symmetry (here no deterministic CA describes Fig-\ref{fig:bwsym1} $(c)$). The aim of the so-called {\em gauging procedure} is to extend a theory order so as to enforce a given local symmetry.
\begin{figure}[ht]
\begin{subfigure}[t]{0.4\textwidth}
  \centering
  \Large
  \resizebox{\textwidth}{!}{\begin{tikzpicture}


\draw[color=gray] (5.5,0) -- (12.5,7);
\draw[color=gray] (2.5,1) -- (8.5,7);
\draw[color=gray] (10,0.5) -- (14,4.5);
\draw[color=gray] (2,4.5) -- (4.5,7);

\draw[color=gray] (4,0.5) -- (2.5,2);
\draw[color=gray] (14,2.5) -- (9.5,7);
\draw[color=gray] (8,0.5) -- (2.5,6);
\draw[color=gray] (12,0.5) -- (5.5,7);

\node at (3, 1.5) [draw,scale=1,aspect=1,diamond,color=gray, fill=white]{}; 
\node at (7, 1.5) [draw,scale=1,aspect=1,diamond,color=gray, fill=white]{}; 
\node at (11, 1.5) [draw,scale=1,aspect=1,diamond,color=gray, fill=white]{}; 

\node at (5, 3.5) [draw,scale=1,aspect=1,diamond,color=gray, fill=white]{}; 
\node at (9, 3.5) [draw,scale=1,aspect=1,diamond,color=gray, fill=white]{}; 
\node at (13, 3.5) [draw,scale=1,aspect=1,diamond,color=gray, fill=white]{}; 

\node at (3, 5.5) [draw,scale=1,aspect=1,diamond,color=gray, fill=white]{}; 
\node at (7, 5.5) [draw,scale=1,aspect=1,diamond,color=gray, fill=white]{}; 
\node at (11, 5.5) [draw,scale=1,aspect=1,diamond,color=gray, fill=white]{};

  \filldraw[color=black, fill=white, thick](4,0) rectangle (5,1);
  \filldraw[color=black, fill=black, thick](5,0) rectangle (6,1);
  \filldraw[color=black, fill=white, thick](8,0) rectangle (9,1);
  \filldraw[color=black, fill=white, thick](9,0) rectangle (10,1);
  \filldraw[color=black, fill=black, thick](12,0) rectangle (13,1);
  \filldraw[color=black, fill=white, thick](13,0) rectangle (14,1);

  \filldraw[color=black, fill=white, thick](2,2) rectangle (3,3);
  \filldraw[color=black, fill=white, thick](3,2) rectangle (4,3);
  \filldraw[color=black, fill=white, thick](6,2) rectangle (7,3);
  \filldraw[color=black, fill=black, thick](7,2) rectangle (8,3);
  \filldraw[color=black, fill=black, thick](10,2) rectangle (11,3);
  \filldraw[color=black, fill=white, thick](11,2) rectangle (12,3);
  
  \filldraw[color=black, fill=white, thick](4,4) rectangle (5,5);
  \filldraw[color=black, fill=white, thick](5,4) rectangle (6,5);
  \filldraw[color=white, fill=black, thick](8,4) rectangle (9,5);
  \filldraw[color=white, fill=black, thick](9,4) rectangle (10,5);
  \filldraw[color=black, fill=white, thick](12,4) rectangle (13,5);
  \filldraw[color=black, fill=white, thick](13,4) rectangle (14,5);

  \filldraw[color=black, fill=white, thick](2,6) rectangle (3,7);
  \filldraw[color=black, fill=white, thick](3,6) rectangle (4,7);
  \filldraw[color=black, fill=black, thick](6,6) rectangle (7,7);
  \filldraw[color=black, fill=white, thick](7,6) rectangle (8,7);
  \filldraw[color=black, fill=white, thick](10,6) rectangle (11,7);
  \filldraw[color=black, fill=black, thick](11,6) rectangle (12,7);


\draw[color=black] (1, 0.5) node {t};
\draw[color=black] (1, 2.5) node {t+1};
\draw[color=black] (1, 4.5) node {t+2};
\draw[color=black] (1, 6.5) node {t+3};

\draw[color=black] (3, -1) node {x-2};
\draw[color=black] (5, -1) node {x-1};
\draw[color=black] (7, -1) node {x};
\draw[color=black] (9, -1) node {x+1};
\draw[color=black] (11, -1) node {x+2};
\draw[color=black] (13, -1) node {x+3};

\end{tikzpicture}}
  \caption{\label{fig:spacetimediagram}An $R$--valid spacetime diagram with two particles moving in opposite directions}
\end{subfigure}\hfill
\begin{subfigure}[t]{0.4\textwidth}
  \centering
  \resizebox{\textwidth}{!}{\begin{tikzpicture}


\draw[color=black, thick] (4.5,1) -- (1,4.5);
\draw[color=black, thick] (1.5,1) -- (5,4.5);

\draw[color=gray, thick] (1,4.5) -- (0,3.5);
\draw[color=gray, thick] (5,4.5) -- (6,3.5);
\draw[color=gray, thick] (0,1.5) -- (0.5,1);
\draw[color=gray, thick] (6,1.5) -- (5.5,1);

\stateTikz{3}{3.2}{1}{0.5}{0}{0};
\draw (3,4) node {$\psi(x,t+1)$};
\draw (3.5,3.45) node {$\psi^+$};
\draw (2.5,3.45) node {$\psi^-$};
\stateTikz{1}{1.2}{1}{0.5}{0}{0};
\draw (1,1.9) node {$\psi(x-1,t)$};
\stateTikz{5}{1.2}{1}{0.5}{0}{0};
\draw (5,1.9) node {$\psi(x+1,t)$};

\node at (3, 2.5) [draw,scale=1,aspect=1,diamond,color=black, fill=white]{W}; 
\node at (1, 4.5) [draw,scale=1,aspect=1,diamond,color=gray, fill=white]{W}; 
\node at (5, 4.5) [draw,scale=1,aspect=1,diamond,color=gray, fill=white]{W}; 
\end{tikzpicture}}
  \caption{\label{fig:framework}Framework of study}
\end{subfigure}
\caption{Representation of the framework of study.}
\end{figure}
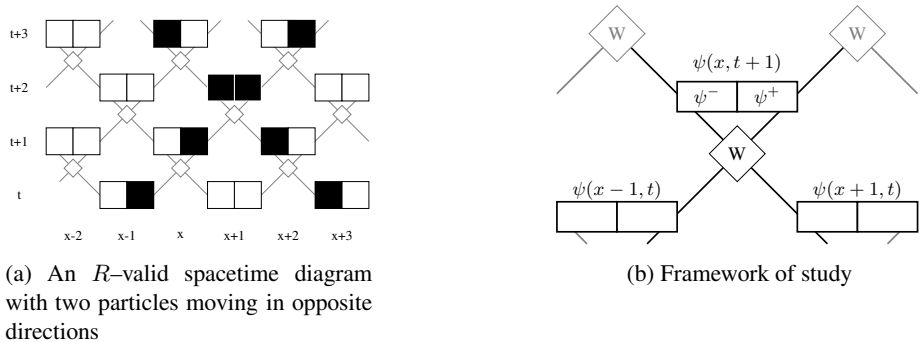

\paragraph{Local transformation and its invariant.} 
In our running example we will be interested in enforcing a {\em local} black/white--symmetry. We formalize this by giving ourselves a bit field $\varphi:\mathbb{Z}^2\rightarrow\{0,1\}\cong \mathbb{Z}_2$ that specifies, at each spacetime point, whether the symmetry is to be applied. In other words, the action of the $\mathbb{Z}_2$ group at $(x, t)$ gets represented upon $\Sigma$ by
\begin{equation}
G_{\varphi}(x,t)=(X\otimes X)^{\varphi(x,t)}
\end{equation}
with $X$ the \textsc{Not} gate. Then, an entire spacetime diagram $\psi$ transforms into an entire spacetime diagram $G_\varphi \psi$ via $(G_\varphi \psi)(x,t) = G_{\varphi}(x,t) \psi(x,t)$. This is the symmetry we are trying to enforce. Thus, whenever two spacetime diagrams are related by a transformation $G_\varphi$ for some $\varphi$, they are understood as {\em physically equivalent}. For instance, in Fig-\ref{fig:bwsym1}, the first three diagrams are physically equivalent with respect to the symmetry. They all represent this one particle moving right, which can be understood as an invariant of the symmetry. Given a spacetime diagram $\psi$, we write $\widetilde{\psi}=\{ G_\varphi\psi\,|\,\varphi\in \mathbb{Z}^2\rightarrow\mathbb{Z}_2\}$ for its {\em invariant}, (physical) equivalence class. In the case of our field $\psi$, the bit field $J(x,t)=\psi^+(x,t)-\psi^-(x,t)$\ (\!\!\!$\mod 2$) fully characterizes $\widetilde{\psi}$, since for all $\psi$ and $\psi'$, $G_\varphi \psi=\psi'$ if and only if $J=J'$. Fig-\ref{subfig:Jfield} shows the underlying $J$.

\begin{figure}[ht]
\begin{subfigure}[t]{0.4\textwidth}
  \centering
  \Large
  \resizebox{\textwidth}{!}{\begin{tikzpicture}

\draw[color=gray] (5.5,0) -- (12.5,7);
\draw[color=gray] (2.5,1) -- (8.5,7);
\draw[color=gray] (10,0.5) -- (14,4.5);
\draw[color=gray] (2,4.5) -- (4.5,7);

\draw[color=gray] (4,0.5) -- (2.5,2);
\draw[color=gray] (14,2.5) -- (9.5,7);
\draw[color=gray] (8,0.5) -- (2.5,6);
\draw[color=gray] (12,0.5) -- (5.5,7);

\node at (3, 1.5) [draw,scale=1,aspect=1,diamond,color=gray, fill=white]{}; 
\node at (7, 1.5) [draw,scale=1,aspect=1,diamond,color=gray, fill=white]{}; 
\node at (11, 1.5) [draw,scale=1,aspect=1,diamond,color=gray, fill=white]{}; 

\node at (5, 3.5) [draw,scale=1,aspect=1,diamond,color=gray, fill=white]{}; 
\node at (9, 3.5) [draw,scale=1,aspect=1,diamond,color=gray, fill=white]{}; 
\node at (13, 3.5) [draw,scale=1,aspect=1,diamond,color=gray, fill=white]{}; 

\node at (3, 5.5) [draw,scale=1,aspect=1,diamond,color=gray, fill=white]{}; 
\node at (7, 5.5) [draw,scale=1,aspect=1,diamond,color=gray, fill=white]{}; 
\node at (11, 5.5) [draw,scale=1,aspect=1,diamond,color=gray, fill=white]{};

  \filldraw[color=black, fill=black, thick](4,0) rectangle (5,1);
  \filldraw[color=black, fill=white, thick](5,0) rectangle (6,1);
  \filldraw[color=white, fill=black, thick](8,0) rectangle (9,1);
  \filldraw[color=white, fill=black, thick](9,0) rectangle (10,1);
  \filldraw[color=white, fill=black, thick](12,0) rectangle (13,1);
  \filldraw[color=white, fill=black, thick](13,0) rectangle (14,1);

  \filldraw[color=white, fill=black, thick](2,2) rectangle (3,3);
  \filldraw[color=white, fill=black, thick](3,2) rectangle (4,3);
  \filldraw[color=black, fill=black, thick](6,2) rectangle (7,3);
  \filldraw[color=black, fill=white, thick](7,2) rectangle (8,3);
  \filldraw[color=white, fill=black, thick](10,2) rectangle (11,3);
  \filldraw[color=white, fill=black, thick](11,2) rectangle (12,3);

  \filldraw[color=white, fill=black, thick](4,4) rectangle (5,5);
  \filldraw[color=white, fill=black, thick](5,4) rectangle (6,5);
  \filldraw[color=black, fill=black, thick](8,4) rectangle (9,5);
  \filldraw[color=black, fill=white, thick](9,4) rectangle (10,5);
  \filldraw[color=white, fill=black, thick](12,4) rectangle (13,5);
  \filldraw[color=white, fill=black, thick](13,4) rectangle (14,5);

  \filldraw[color=white, fill=black, thick](2,6) rectangle (3,7);
  \filldraw[color=white, fill=black, thick](3,6) rectangle (4,7);
  \filldraw[color=white, fill=black, thick](6,6) rectangle (7,7);
  \filldraw[color=white, fill=black, thick](7,6) rectangle (8,7);
  \filldraw[color=black, fill=black, thick](10,6) rectangle (11,7);
  \filldraw[color=black, fill=white, thick](11,6) rectangle (12,7);


\draw[color=black] (1, 0.5) node {t};
\draw[color=black] (1, 2.5) node {t+1};
\draw[color=black] (1, 4.5) node {t+2};
\draw[color=black] (1, 6.5) node {t+3};

\draw[color=black] (3, -1) node {x-2};
\draw[color=black] (5, -1) node {x-1};
\draw[color=black] (7, -1) node {x};
\draw[color=black] (9, -1) node {x+1};
\draw[color=black] (11, -1) node {x+2};
\draw[color=black] (13, -1) node {x+3};

\end{tikzpicture}}
  \caption{$R$--valid spacetime diagram showing a particle moving right.}
\end{subfigure}\hfill
\begin{subfigure}[t]{0.4\textwidth}
  \centering
  \Large
  \resizebox{\textwidth}{!}{\begin{tikzpicture}


\draw[color=gray] (5.5,0) -- (12.5,7);
\draw[color=gray] (2.5,1) -- (8.5,7);
\draw[color=gray] (10,0.5) -- (14,4.5);
\draw[color=gray] (2,4.5) -- (4.5,7);

\draw[color=gray] (4,0.5) -- (2.5,2);
\draw[color=gray] (14,2.5) -- (9.5,7);
\draw[color=gray] (8,0.5) -- (2.5,6);
\draw[color=gray] (12,0.5) -- (5.5,7);

\node at (3, 1.5) [draw,scale=1,aspect=1,diamond,color=gray, fill=white]{}; 
\node at (7, 1.5) [draw,scale=1,aspect=1,diamond,color=gray, fill=white]{}; 
\node at (11, 1.5) [draw,scale=1,aspect=1,diamond,color=gray, fill=white]{}; 

\node at (5, 3.5) [draw,scale=1,aspect=1,diamond,color=gray, fill=white]{}; 
\node at (9, 3.5) [draw,scale=1,aspect=1,diamond,color=gray, fill=white]{}; 
\node at (13, 3.5) [draw,scale=1,aspect=1,diamond,color=gray, fill=white]{}; 

\node at (3, 5.5) [draw,scale=1,aspect=1,diamond,color=gray, fill=white]{}; 
\node at (7, 5.5) [draw,scale=1,aspect=1,diamond,color=gray, fill=white]{}; 
\node at (11, 5.5) [draw,scale=1,aspect=1,diamond,color=gray, fill=white]{};

  \filldraw[color=black, fill=white, thick](4,0) rectangle (5,1);
  \filldraw[color=black, fill=black, thick](5,0) rectangle (6,1);
  \filldraw[color=black, fill=white, thick](8,0) rectangle (9,1);
  \filldraw[color=black, fill=white, thick](9,0) rectangle (10,1);
  \filldraw[color=black, fill=white, thick](12,0) rectangle (13,1);
  \filldraw[color=black, fill=white, thick](13,0) rectangle (14,1);

  \filldraw[color=black, fill=white, thick](2,2) rectangle (3,3);
  \filldraw[color=black, fill=white, thick](3,2) rectangle (4,3);
  \filldraw[color=black, fill=white, thick](6,2) rectangle (7,3);
  \filldraw[color=black, fill=black, thick](7,2) rectangle (8,3);
  \filldraw[color=black, fill=white, thick](10,2) rectangle (11,3);
  \filldraw[color=black, fill=white, thick](11,2) rectangle (12,3);
    
  \filldraw[color=black, fill=white, thick](4,4) rectangle (5,5);
  \filldraw[color=black, fill=white, thick](5,4) rectangle (6,5);
  \filldraw[color=black, fill=white, thick](8,4) rectangle (9,5);
  \filldraw[color=black, fill=black, thick](9,4) rectangle (10,5);
  \filldraw[color=black, fill=white, thick](12,4) rectangle (13,5);
  \filldraw[color=black, fill=white, thick](13,4) rectangle (14,5);

  \filldraw[color=black, fill=white, thick](2,6) rectangle (3,7);
  \filldraw[color=black, fill=white, thick](3,6) rectangle (4,7);
  \filldraw[color=black, fill=white, thick](6,6) rectangle (7,7);
  \filldraw[color=black, fill=white, thick](7,6) rectangle (8,7);
  \filldraw[color=black, fill=white, thick](10,6) rectangle (11,7);
  \filldraw[color=black, fill=black, thick](11,6) rectangle (12,7);


\draw[color=black] (1, 0.5) node {t};
\draw[color=black] (1, 2.5) node {t+1};
\draw[color=black] (1, 4.5) node {t+2};
\draw[color=black] (1, 6.5) node {t+3};

\draw[color=black] (3, -1) node {x-2};
\draw[color=black] (5, -1) node {x-1};
\draw[color=black] (7, -1) node {x};
\draw[color=black] (9, -1) node {x+1};
\draw[color=black] (11, -1) node {x+2};
\draw[color=black] (13, -1) node {x+3};

\end{tikzpicture}}
  \caption{Still an $R$--valid spacetime diagram after applying the global symmetry $G_\varphi$ with $\varphi$ constant equal to one.}
\end{subfigure}

\begin{subfigure}[t]{0.4\textwidth}
  \centering
  \Large
  \resizebox{\textwidth}{!}{\begin{tikzpicture}

\draw[color=gray] (5.5,0) -- (12.5,7);
\draw[color=gray] (2.5,1) -- (8.5,7);
\draw[color=gray] (10,0.5) -- (14,4.5);
\draw[color=gray] (2,4.5) -- (4.5,7);

\draw[color=gray] (4,0.5) -- (2.5,2);
\draw[color=gray] (14,2.5) -- (9.5,7);
\draw[color=gray] (8,0.5) -- (2.5,6);
\draw[color=gray] (12,0.5) -- (5.5,7);

\node at (3, 1.5) [draw,scale=1,aspect=1,diamond,color=gray, fill=white]{}; 
\node at (7, 1.5) [draw,scale=1,aspect=1,diamond,color=gray, fill=white]{}; 
\node at (11, 1.5) [draw,scale=1,aspect=1,diamond,color=gray, fill=white]{}; 

\node at (5, 3.5) [draw,scale=1,aspect=1,diamond,color=gray, fill=white]{}; 
\node at (9, 3.5) [draw,scale=1,aspect=1,diamond,color=gray, fill=white]{}; 
\node at (13, 3.5) [draw,scale=1,aspect=1,diamond,color=gray, fill=white]{}; 

\node at (3, 5.5) [draw,scale=1,aspect=1,diamond,color=gray, fill=white]{}; 
\node at (7, 5.5) [draw,scale=1,aspect=1,diamond,color=gray, fill=white]{}; 
\node at (11, 5.5) [draw,scale=1,aspect=1,diamond,color=gray, fill=white]{};

  \filldraw[color=black, fill=black, thick](4,0) rectangle (5,1);
  \filldraw[color=black, fill=white, thick](5,0) rectangle (6,1);
  \filldraw[color=black, fill=white, thick](8,0) rectangle (9,1);
  \filldraw[color=black, fill=white, thick](9,0) rectangle (10,1);
  \filldraw[color=black, fill=white, thick](12,0) rectangle (13,1);
  \filldraw[color=black, fill=white, thick](13,0) rectangle (14,1);

  \filldraw[color=white, fill=black, thick](2,2) rectangle (3,3);
  \filldraw[color=white, fill=black, thick](3,2) rectangle (4,3);
  \filldraw[color=black, fill=black, thick](6,2) rectangle (7,3);
  \filldraw[color=black, fill=white, thick](7,2) rectangle (8,3);
  \filldraw[color=black, fill=white, thick](10,2) rectangle (11,3);
  \filldraw[color=black, fill=white, thick](11,2) rectangle (12,3);

  \filldraw[color=white, fill=black, thick](4,4) rectangle (5,5);
  \filldraw[color=white, fill=black, thick](5,4) rectangle (6,5);
  \filldraw[color=black, fill=white, thick](8,4) rectangle (9,5);
  \filldraw[color=black, fill=black, thick](9,4) rectangle (10,5);
  \filldraw[color=black, fill=white, thick](12,4) rectangle (13,5);
  \filldraw[color=black, fill=white, thick](13,4) rectangle (14,5);

  \filldraw[color=white, fill=black, thick](2,6) rectangle (3,7);
  \filldraw[color=white, fill=black, thick](3,6) rectangle (4,7);
  \filldraw[color=white, fill=black, thick](6,6) rectangle (7,7);
  \filldraw[color=white, fill=black, thick](7,6) rectangle (8,7);
  \filldraw[color=black, fill=white, thick](10,6) rectangle (11,7);
  \filldraw[color=black, fill=black, thick](11,6) rectangle (12,7);


\draw[color=blue, dashed, very thick] (8,-1.5) -- (8,7.2);
\draw[color=blue] (5, -1) node {\LARGE $x\leq0$, $\varphi(x,.) = 0$};
\draw[color=blue] (11, -1) node {\LARGE $x>0$, $\varphi(x,.) = 1$};

\draw[color=black] (1, 0.5) node {t};
\draw[color=black] (1, 2.5) node {t+1};
\draw[color=black] (1, 4.5) node {t+2};
\draw[color=black] (1, 6.5) node {t+3};

\end{tikzpicture}}
  \caption{\label{subfig:brokensym}Not an $R$--valid spacetime diagram after applying the local symmetry $G_\varphi$ with space-dependent $\varphi$.}
\end{subfigure}\hfill
\begin{subfigure}[t]{0.4\textwidth}
  \centering
  \Large
  \resizebox{\textwidth}{!}{\begin{tikzpicture}


\draw[color=gray] (5.5,0) -- (12.5,7);
\draw[color=gray] (2.5,1) -- (8.5,7);
\draw[color=gray] (10,0.5) -- (14,4.5);
\draw[color=gray] (2,4.5) -- (4.5,7);

\draw[color=gray] (4,0.5) -- (2.5,2);
\draw[color=gray] (14,2.5) -- (9.5,7);
\draw[color=gray] (8,0.5) -- (2.5,6);
\draw[color=gray] (12,0.5) -- (5.5,7);

\node at (3, 1.5) [draw,scale=1,aspect=1,diamond,color=gray, fill=white]{}; 
\node at (7, 1.5) [draw,scale=1,aspect=1,diamond,color=gray, fill=white]{}; 
\node at (11, 1.5) [draw,scale=1,aspect=1,diamond,color=gray, fill=white]{}; 

\node at (5, 3.5) [draw,scale=1,aspect=1,diamond,color=gray, fill=white]{}; 
\node at (9, 3.5) [draw,scale=1,aspect=1,diamond,color=gray, fill=white]{}; 
\node at (13, 3.5) [draw,scale=1,aspect=1,diamond,color=gray, fill=white]{}; 

\node at (3, 5.5) [draw,scale=1,aspect=1,diamond,color=gray, fill=white]{}; 
\node at (7, 5.5) [draw,scale=1,aspect=1,diamond,color=gray, fill=white]{}; 
\node at (11, 5.5) [draw,scale=1,aspect=1,diamond,color=gray, fill=white]{};

  \filldraw[color=black, fill=black, thick](4,0) rectangle (6,1);
  \filldraw[color=black, fill=white, thick](8,0) rectangle (10,1);
  \filldraw[color=black, fill=white, thick](12,0) rectangle (14,1);

  \filldraw[color=black, fill=white, thick](2,2) rectangle (4,3);
  \filldraw[color=black, fill=black, thick](6,2) rectangle (8,3);
  \filldraw[color=black, fill=white, thick](10,2) rectangle (12,3);
    
  \filldraw[color=black, fill=white, thick](4,4) rectangle (6,5);
  \filldraw[color=black, fill=black, thick](8,4) rectangle (10,5);
  \filldraw[color=black, fill=white, thick](12,4) rectangle (14,5);

  \filldraw[color=black, fill=white, thick](2,6) rectangle (4,7);
  \filldraw[color=black, fill=white, thick](6,6) rectangle (8,7);
  \filldraw[color=black, fill=black, thick](10,6) rectangle (12,7);


\draw[color=black] (1, 0.5) node {t};
\draw[color=black] (1, 2.5) node {t+1};
\draw[color=black] (1, 4.5) node {t+2};
\draw[color=black] (1, 6.5) node {t+3};

\draw[color=black] (3, -1) node {x-2};
\draw[color=black] (5, -1) node {x-1};
\draw[color=black] (7, -1) node {x};
\draw[color=black] (9, -1) node {x+1};
\draw[color=black] (11, -1) node {x+2};
\draw[color=black] (13, -1) node {x+3};

\end{tikzpicture}}
  \caption{\label{subfig:Jfield}The $J$-field that characterizes the invariant under $G_\varphi$, common to the other three spacetime diagrams.}
\end{subfigure}
\caption{\label{fig:bwsym1}Three physically equivalent spacetime diagrams, and their invariant.}
\end{figure}
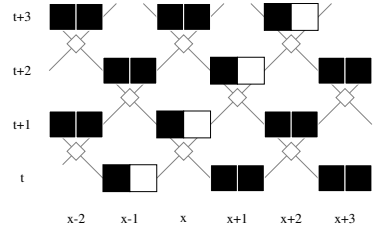
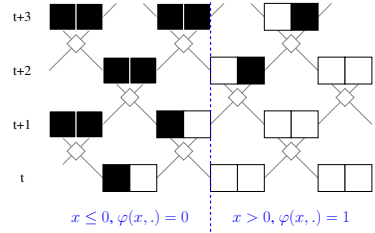
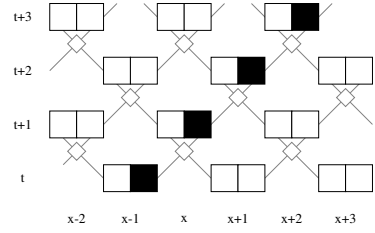
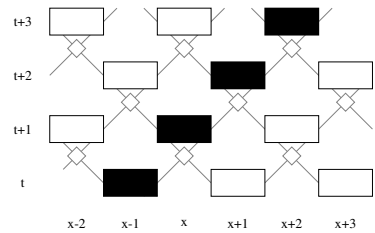

\paragraph{The gauge invariance condition.} 
Given $\psi(.,t)$ and $(G_{\varphi}\psi)(.,t)$ two physically equivalent inputs, it should be the case that our theory produces two physically equivalent outputs $\psi(.,t+1)$ and $(G_{\varphi}\psi)(.,t+1)$. This leads to the following definition.
\begin{definition}[Gauge-invariance]
A theory $T$ is gauge-invariant if and only if there exists $Z$ a theory such that for all $\varphi$
\begin{align}\label{eq:gaugeinvariance}
G_{Z\varphi} \circ T=T\circ G_{\varphi}
\end{align}
\end{definition}

The above-defined RCA fails to meet this requirement. 
An example of this failure is provided by Fig-\ref{fig:bwsym1}, which shows three physically equivalent spacetime diagrams, i.e. that are $G_\varphi$--related. Clearly the first two are $R$--valid, but the third one is not, as can be seen from looking at $\psi(x,t+1)$.\\
Indeed, on the one hand cell $\psi(x,t+1)$ of Fig-\ref{subfig:brokensym} needs have different-color subcells, as it is a $G_{\varphi}(x,t+1)$ of that of the other diagrams, and $G_{\varphi}(x,t+1)$ conserves same-colorness. But, on the other hand, cell $\psi(x,t+1)$ of Fig-\ref{subfig:brokensym} needs have same-color subcells, as it is produced by a $W$ which is fed with same-color subcells---due to the particular choice of $\varphi(x-1,t)$ and $\varphi(x+1,t)$---and since $W$ conserves same-colorness. This cannot be fixed with a better choice of $\varphi(x,t+1)$. Therefore, our previously defined RCA fails to verify the gauge-invariance condition.\\ 
The gauging procedure proceeds by extending $R$ into an inhomogeneous dynamics.

\section{The gauge field}\label{sec:gaugefield}
\paragraph{Introducing the gauge field.} 
In order to obtain the gauge-invariance condition (\ref{eq:gaugeinvariance}), the standard procedure is to make the theory $R$ into an inhomogeneous theory $R_\bullet$, and aim at {\em inhomogeneous gauge-invariance}:
\begin{definition}[Inhomogeneous gauge-invariance]
A theory $T_\bullet$ is {\em inhomogeneous gauge-invariant} if and only if there exists $Z$ a theory such that for all $Z$--valid $\varphi$,
\begin{align}\label{eq:inhomoggaugeinvariance}
G_{Z\varphi} \circ T_\bullet=T_{G_\varphi \bullet}\circ G_{\varphi}
\end{align}
\end{definition}
The spacetime diagram $A$, which specifies the spacetime dependency of $R$, is referred to as the {\em gauge field}, or Ehresmann connection in mathematics. Back to our running example, we are thus looking for an extension of $R$ into an $R_A$ and of $G_\varphi$, so that it acts on both $\psi$ and $A$, in order to achieve condition \eqref{eq:inhomoggaugeinvariance}. Developing, we need that there exists $Z$ such that for all $Z$--valid $\varphi$, for all $A$, for all $R_A$--valid $\psi$, for all $x$, 
\begin{align}
(G_{Z\varphi}(R_A \psi))(x)&=(R_{G_{\varphi}A}(G_{\varphi}\psi))(x)
\end{align}
which, for our running example, translates into :
\begin{align}
(X \otimes X)^{(Z\varphi)(x)} &(W_A (\psi^+(x-1)\otimes\psi^-(x+1)))\\
&= W_{G_\varphi A} (X^{\varphi(x-1)}\psi^+(x-1)\otimes X^{\varphi(x+1)}\psi^-(x+1))
\end{align}
This is equivalent to
\begin{equation}
W_{G_\varphi A}  = (X \otimes X)^{(Z\varphi)(x)} W_{A} (X^{-\varphi(x-1)} \otimes X^{-\varphi(x+1)}).
\end{equation}
A somewhat minimal choice verifying the above condition is to take  $A:\mathbb{Z}^2\rightarrow\mathbb{Z}_2^2$ a 2-bits field, and 
\begin{equation}
W_A = W (X^{A_r} \otimes X^{A_l}),
\end{equation}
with $A$ transforming under $G_\varphi$ as :

\begin{equation}
A(x) = \begin{pmatrix} A_r(x) \\ A_l(x)
\end{pmatrix}\mapsto
\begin{pmatrix}
A_r(x) + (Z\varphi)(x)-\varphi(x-1) \\
A_l(x) + (Z\varphi)(x)-\varphi(x+1)
\end{pmatrix} = (G_\varphi A)(x)
\end{equation}

Indeed, 
\begin{align}
W_{(G_\varphi A) (x)}&=W (X^{A_r(x)+ (Z\varphi)(x)-\varphi(x-1)} \otimes X^{A_l(x)+ (Z\varphi)(x)-\varphi(x+1)})\\
&=W(X \otimes X)^{(Z\varphi)(x)}  (X^{A_r(x)} \otimes X^{A_l(x))}) (X^{-\varphi(x-1)} \otimes X^{-\varphi(x+1)})\\
&=(X \otimes X)^{(Z\varphi)(x)} W_A (X^{-\varphi(x-1)} \otimes X^{-\varphi(x+1)}).
\end{align}
It follows that the induced $R_\bullet$ verifies the inhomogeneous gauge-invariance condition \eqref{eq:inhomoggaugeinvariance}. This procedure is reminiscent of the route physics follows to account for a local phase transformation on the state vector $\psi(x,t)$, which leads to the modern formulation of electrodynamics, with $A(x,t)$ playing the role of the electromagnetic potential. 

\paragraph{Invariant of the gauge field.\label{par:gaugefieldinvariant}} Since $A$ also transforms under $G_\varphi$, we may again seek to characterize its invariant $\widetilde{A}=\{ G_\varphi A\,|\,\varphi\in \mathbb{Z}^2\rightarrow\mathbb{Z}_2\}$ by means of some field $F$. This time, in order to do so, we introduce the light-like discrete derivatives  
\begin{align}
\Delta_r A(x,t) &= A(x,t+1) - A(x-1,t) \\
\Delta_l A(x,t) &= A(x,t+1) - A(x+1,t).
\end{align}
\begin{remark}[Gauge-field invariant] \label{prop:gaugeinv} The bit field
$F(x,t) = \Delta_r A_l(x,t) -\Delta_l A_r(x,t)$ fully characterizes the invariant of the gauge field. More precisely, for any $A$ and $A'$, there exists $\varphi$ such that $G_\varphi A=A'$ is equivalent to $F=F'$.
\end{remark}
\begin{proof}
A lengthy but easy computation shows that, given any $A$ and $A'$, $G_\varphi A=A'$ entails that $F=F'$. The converse is harder to prove, but also true. Indeed, suppose that we are given $A$ and $A'$ such that $F=F'$. We want to construct a $\varphi$ such that $G_\varphi A=A'$, i.e. such that we have both 
\begin{align}
\Delta_r \varphi = A'_r-A_r \qquad\textrm{and }\qquad\Delta_l \varphi = A'_l-A_l.\label{eq:phireq}
\end{align}
Clearly, starting from an initial spacelike configuration $\varphi$, the requirements \eqref{eq:phireq} fix the rest of $\varphi$ across spacetime. Unless they conflict. This could happen every time we close up a square. Starting from $\varphi(x,t)$, say, the requirements \eqref{eq:phireq} provide two prescriptions for $\varphi(x,t+2)$, namely $\varphi(x,t)+(A'_l-A_l)(x-1,t)+(A'_r-A_r)(x,t+1)$ via the left-then-right path, and $\varphi(x,t)+(A'_r-A_r)(x+1,t)+(A'_l-A_l)(x,t+1)$ via the right-then-left path. These need be equal, i.e we need 
\begin{align} 
(A'_l-A_l)(x,t+1)\!-\!(A'_l-A_l)(x-1,t) &=(A'_r-A_r)(x,t+1)\!-\!(A'_r-A_r)(x+1,t)\\
\Delta_r (A'_l-A_l)(x,t) &=\Delta_l (A'_r-A_r)(x,t)\\
\Delta_r A'_l- \Delta_r A_l &=\Delta_l A'_r- \Delta_l A_r\\
\Delta_l A_r-\Delta_r A_l &=\Delta_l A'_r-\Delta_r A'_l\\
F(x,t) &= F'(x,t)
\end{align}
which is our hypothesis. It follows that $\varphi$ exists and so the converse holds. $F$ fully characterizes $\widetilde{A}$. \hfill$\square$
\end{proof}
The role played by this discrete bit-field $F$ is analogous to that of the electromagnetic tensor, a differential $2$--form, which is the exterior derivative of the electromagnetic potential $A(x,t)$ and whose derivatives are prescribed by the Maxwell equations.

\begin{figure}
\centering
\resizebox{0.45\textwidth}{!}{\begin{tikzpicture}


\draw[color=black, thick] (6,-0.5) -- (-1,4.5);
\draw[color=black, thick] (0,-0.5) -- (7,4.5);

\draw[color=gray, thick] (-1,4.5) -- (-2,4);
\draw[color=gray, thick] (7,4.5) -- (8,4);

\stateTikz{3}{2.5}{1.6}{0.5}{0}{0};
\draw (3,3.25) node {$\psi(x,t+1)$};
\draw (2.6,2.75) node {$\psi^-$};
\draw (3.6,2.75) node {$\psi^+$};
\stateTikz{-1}{-0.5}{1.6}{0.5}{0}{0};
\draw (-1,0.3) node {$\psi(x-1,t)$};
\stateTikz{7}{-0.5}{1.6}{0.5}{0}{0};
\draw (7,0.3) node {$\psi(x+1,t)$};

\node at (3, 1.5) [draw,scale=1.2,aspect=2,diamond,color=black, fill=white]{$W_A(x,t)$}; 
\node at (-1, 4.5) [draw,scale=0.8,aspect=2,diamond,color=gray, fill=white]{$W_A(x-1,t+1)$}; 
\node at (7, 4.5) [draw,scale=0.8,aspect=2,diamond,color=gray, fill=white]{$W_A(x+1,t+1)$}; 

\ovalTikz{0.8}{0.3}{0}{0}{3}{-0.25};
\draw[color=\gaugecol] (3,0.3) node {$A(x,t)$};
\draw[color=\gaugecol] (2.6,-0.25) node {$A_l$};
\draw[color=\gaugecol] (3.4,-0.25) node {$A_r$};
\ovalTikz{0.8}{0.3}{0}{0}{-1}{2.75};
\draw[color=\gaugecol] (-1,3.25) node {$A(x-1,t+1)$};
\ovalTikz{0.8}{0.3}{0}{0}{7}{2.75};
\draw[color=\gaugecol] (7,3.25) node {$A(x+1,t+1)$};
\end{tikzpicture}

}
\caption{\label{fig:framework2}The extended theory $R_\bullet$ now depends on a $2$--bits field $A$.}
\end{figure}
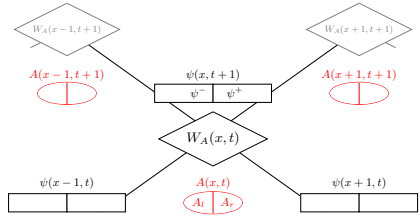
\paragraph{Gauge field physics.} 
It is crucial to understand that, even though $A$ was introduced just to enforce a symmetry, i.e. to make sure that physically equivalent states are mapped into physically equivalent states\ldots this newly introduced $A$ is also capable of a range of other things, i.e. it produces new physics. For instance, Fig-\ref{fig:gaugeinfluence} shows how, starting from the same initial conditions for $J$, but choosing different initial conditions for $A$, can lead to rather different $R_A$-valid spacetime diagrams---which may (Fig-\ref{fig:gauge5}) or may not be (Fig-\ref{fig:gauge3}) related by a $G_\phi$.\\
At this stage we can have two points of view upon $A$. Either $A$ is seen as an independent field, which could be to some extent tuned by the user/experimentalist (\eg\  in the case of the electromagnetism, one can physically engineer each component of the electromagnetic tensor, $F$, namely the electric and the magnetic field). Or we must extend the configuration space so as to account for $A$, as suggested by Fig-\ref{fig:framework2}. Of course if we do that we need to provide a dynamics for $A$, \ie\ we need to look for a theory $T$ upon $c(x,t)=(\psi(x,t),A(x,t))$, which still verifies the gauge-invariance condition \eqref{eq:gaugeinvariance}.

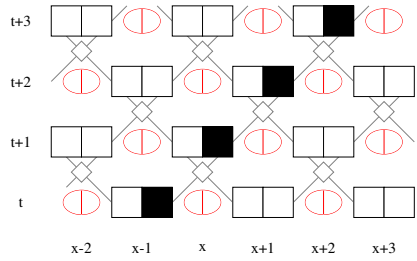
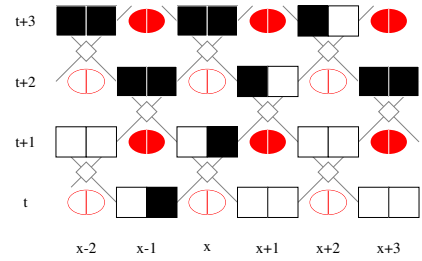
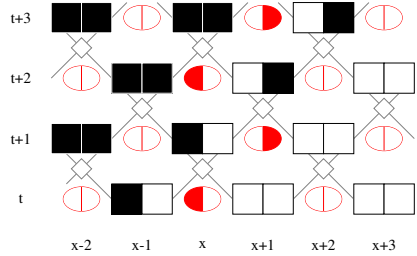
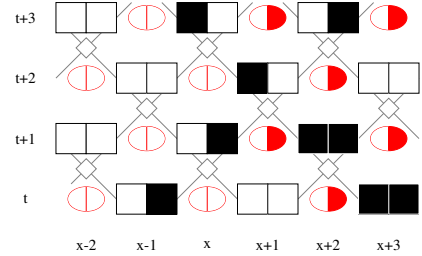
\begin{figure}
\begin{subfigure}[t]{0.45\textwidth}
  \centering
  \Large
  \resizebox{\textwidth}{!}{\begin{tikzpicture}


\draw[color=gray] (5.5,0) -- (12.5,7);
\draw[color=gray] (2.5,1) -- (8.5,7);
\draw[color=gray] (10,0.5) -- (14,4.5);
\draw[color=gray] (2,4.5) -- (4.5,7);

\draw[color=gray] (4,0.5) -- (2.5,2);
\draw[color=gray] (14,2.5) -- (9.5,7);
\draw[color=gray] (8,0.5) -- (2.5,6);
\draw[color=gray] (12,0.5) -- (5.5,7);

\node at (3, 1.5) [draw,scale=1,aspect=1,diamond,color=gray, fill=white]{}; 
\node at (7, 1.5) [draw,scale=1,aspect=1,diamond,color=gray, fill=white]{}; 
\node at (11, 1.5) [draw,scale=1,aspect=1,diamond,color=gray, fill=white]{}; 

\node at (5, 3.5) [draw,scale=1,aspect=1,diamond,color=gray, fill=white]{}; 
\node at (9, 3.5) [draw,scale=1,aspect=1,diamond,color=gray, fill=white]{}; 
\node at (13, 3.5) [draw,scale=1,aspect=1,diamond,color=gray, fill=white]{}; 

\node at (3, 5.5) [draw,scale=1,aspect=1,diamond,color=gray, fill=white]{}; 
\node at (7, 5.5) [draw,scale=1,aspect=1,diamond,color=gray, fill=white]{}; 
\node at (11, 5.5) [draw,scale=1,aspect=1,diamond,color=gray, fill=white]{};

  \stateTikz{5}{0}{1}{1}{0}{1};
  \stateTikz{9}{0}{1}{1}{0}{0};
  \stateTikz{13}{0}{1}{1}{0}{0};
  
  \ovalTikz{\gfwid}{\gfhei}{0}{0}{3}{0.5};
  \ovalTikz{\gfwid}{\gfhei}{0}{0}{7}{0.5};
  \ovalTikz{\gfwid}{\gfhei}{0}{0}{11}{0.5};

  \stateTikz{3}{2}{1}{1}{0}{0};
  \stateTikz{7}{2}{1}{1}{0}{1};
  \stateTikz{11}{2}{1}{1}{0}{0};
  
  \ovalTikz{\gfwid}{\gfhei}{0}{0}{5}{2.5};
  \ovalTikz{\gfwid}{\gfhei}{0}{0}{9}{2.5};
  \ovalTikz{\gfwid}{\gfhei}{0}{0}{13}{2.5};
    
  \stateTikz{5}{4}{1}{1}{0}{0};
  \stateTikz{9}{4}{1}{1}{0}{1};
  \stateTikz{13}{4}{1}{1}{0}{0};

  \ovalTikz{\gfwid}{\gfhei}{0}{0}{3}{4.5};
  \ovalTikz{\gfwid}{\gfhei}{0}{0}{7}{4.5};
  \ovalTikz{\gfwid}{\gfhei}{0}{0}{11}{4.5};

  \stateTikz{3}{6}{1}{1}{0}{0};
  \stateTikz{7}{6}{1}{1}{0}{0};
  \stateTikz{11}{6}{1}{1}{0}{1};
  
  \ovalTikz{\gfwid}{\gfhei}{0}{0}{5}{6.5};
  \ovalTikz{\gfwid}{\gfhei}{0}{0}{9}{6.5};
  \ovalTikz{\gfwid}{\gfhei}{0}{0}{13}{6.5};


\draw[color=black] (1, 0.5) node {t};
\draw[color=black] (1, 2.5) node {t+1};
\draw[color=black] (1, 4.5) node {t+2};
\draw[color=black] (1, 6.5) node {t+3};

\draw[color=black] (3, -1) node {x-2};
\draw[color=black] (5, -1) node {x-1};
\draw[color=black] (7, -1) node {x};
\draw[color=black] (9, -1) node {x+1};
\draw[color=black] (11, -1) node {x+2};
\draw[color=black] (13, -1) node {x+3};

\end{tikzpicture}}
  \caption{\label{fig:gauge1}A spacetime diagram with $S_\psi = \mathbf{I}$ ($A(x,t+1)=A_r(x+1,t) \otimes A_l(x-1,t)$) and thus $\widetilde{S}_J = \mathbf{I}$, initialized at $F=0$.}
\end{subfigure}\hfill
\begin{subfigure}[t]{0.45\textwidth}
  \centering
  \Large
  \resizebox{\textwidth}{!}{\begin{tikzpicture}


\draw[color=gray] (5.5,0) -- (12.5,7);
\draw[color=gray] (2.5,1) -- (8.5,7);
\draw[color=gray] (10,0.5) -- (14,4.5);
\draw[color=gray] (2,4.5) -- (4.5,7);

\draw[color=gray] (4,0.5) -- (2.5,2);
\draw[color=gray] (14,2.5) -- (9.5,7);
\draw[color=gray] (8,0.5) -- (2.5,6);
\draw[color=gray] (12,0.5) -- (5.5,7);

\node at (3, 1.5) [draw,scale=1,aspect=1,diamond,color=gray, fill=white]{}; 
\node at (7, 1.5) [draw,scale=1,aspect=1,diamond,color=gray, fill=white]{}; 
\node at (11, 1.5) [draw,scale=1,aspect=1,diamond,color=gray, fill=white]{}; 

\node at (5, 3.5) [draw,scale=1,aspect=1,diamond,color=gray, fill=white]{}; 
\node at (9, 3.5) [draw,scale=1,aspect=1,diamond,color=gray, fill=white]{}; 
\node at (13, 3.5) [draw,scale=1,aspect=1,diamond,color=gray, fill=white]{}; 

\node at (3, 5.5) [draw,scale=1,aspect=1,diamond,color=gray, fill=white]{}; 
\node at (7, 5.5) [draw,scale=1,aspect=1,diamond,color=gray, fill=white]{}; 
\node at (11, 5.5) [draw,scale=1,aspect=1,diamond,color=gray, fill=white]{};

  \stateTikz{5}{0}{1}{1}{0}{1};
  \stateTikz{9}{0}{1}{1}{0}{0};
  \stateTikz{13}{0}{1}{1}{0}{0};
  
  \ovalTikz{\gfwid}{\gfhei}{0}{0}{3}{0.5};
  \ovalTikz{\gfwid}{\gfhei}{0}{0}{7}{0.5};
  \ovalTikz{\gfwid}{\gfhei}{0}{0}{11}{0.5};

  \stateTikz{3}{2}{1}{1}{0}{0};
  \stateTikz{7}{2}{1}{1}{0}{1};
  \stateTikz{11}{2}{1}{1}{0}{0};
  
  \ovalTikz{\gfwid}{\gfhei}{1}{1}{5}{2.5};
  \ovalTikz{\gfwid}{\gfhei}{1}{1}{9}{2.5};
  \ovalTikz{\gfwid}{\gfhei}{1}{1}{13}{2.5};
    
  \stateTikz{5}{4}{1}{1}{1}{1};
  \stateTikz{9}{4}{1}{1}{1}{0};
  \stateTikz{13}{4}{1}{1}{1}{1};

  \ovalTikz{\gfwid}{\gfhei}{0}{0}{3}{4.5};
  \ovalTikz{\gfwid}{\gfhei}{0}{0}{7}{4.5};
  \ovalTikz{\gfwid}{\gfhei}{0}{0}{11}{4.5};

  \stateTikz{3}{6}{1}{1}{1}{1};
  \stateTikz{7}{6}{1}{1}{1}{1};
  \stateTikz{11}{6}{1}{1}{1}{0};
  
  \ovalTikz{\gfwid}{\gfhei}{1}{1}{5}{6.5};
  \ovalTikz{\gfwid}{\gfhei}{1}{1}{9}{6.5};
  \ovalTikz{\gfwid}{\gfhei}{1}{1}{13}{6.5};


\draw[color=black] (1, 0.5) node {t};
\draw[color=black] (1, 2.5) node {t+1};
\draw[color=black] (1, 4.5) node {t+2};
\draw[color=black] (1, 6.5) node {t+3};

\draw[color=black] (3, -1) node {x-2};
\draw[color=black] (5, -1) node {x-1};
\draw[color=black] (7, -1) node {x};
\draw[color=black] (9, -1) node {x+1};
\draw[color=black] (11, -1) node {x+2};
\draw[color=black] (13, -1) node {x+3};

\end{tikzpicture}}
  \caption{\label{fig:gauge2}A spacetime diagram with $S_\psi = (X\otimes X)$ inducing $\widetilde{S}_J = \mathbf{I}$ again, initialized at $F=0$. Observe that the $J$ field is as in $(a)$.}
\end{subfigure}\vspace{0.5cm}
\begin{subfigure}[t]{0.45\textwidth}
  \centering
  \Large
  \resizebox{\textwidth}{!}{\begin{tikzpicture}


\draw[color=gray] (5.5,0) -- (12.5,7);
\draw[color=gray] (2.5,1) -- (8.5,7);
\draw[color=gray] (10,0.5) -- (14,4.5);
\draw[color=gray] (2,4.5) -- (4.5,7);

\draw[color=gray] (4,0.5) -- (2.5,2);
\draw[color=gray] (14,2.5) -- (9.5,7);
\draw[color=gray] (8,0.5) -- (2.5,6);
\draw[color=gray] (12,0.5) -- (5.5,7);

\node at (3, 1.5) [draw,scale=1,aspect=1,diamond,color=gray, fill=white]{}; 
\node at (7, 1.5) [draw,scale=1,aspect=1,diamond,color=gray, fill=white]{}; 
\node at (11, 1.5) [draw,scale=1,aspect=1,diamond,color=gray, fill=white]{}; 

\node at (5, 3.5) [draw,scale=1,aspect=1,diamond,color=gray, fill=white]{}; 
\node at (9, 3.5) [draw,scale=1,aspect=1,diamond,color=gray, fill=white]{}; 
\node at (13, 3.5) [draw,scale=1,aspect=1,diamond,color=gray, fill=white]{}; 

\node at (3, 5.5) [draw,scale=1,aspect=1,diamond,color=gray, fill=white]{}; 
\node at (7, 5.5) [draw,scale=1,aspect=1,diamond,color=gray, fill=white]{}; 
\node at (11, 5.5) [draw,scale=1,aspect=1,diamond,color=gray, fill=white]{};

  \stateTikz{5}{0}{1}{1}{1}{0};
  \stateTikz{9}{0}{1}{1}{0}{0};
  \stateTikz{13}{0}{1}{1}{0}{0};
  
  \ovalTikz{\gfwid}{\gfhei}{0}{0}{3}{0.5};
  \ovalTikz{\gfwid}{\gfhei}{1}{0}{7}{0.5};
  \ovalTikz{\gfwid}{\gfhei}{0}{0}{11}{0.5};

  \stateTikz{3}{2}{1}{1}{1}{1};
  \stateTikz{7}{2}{1}{1}{1}{0};
  \stateTikz{11}{2}{1}{1}{0}{0};
  
  \ovalTikz{\gfwid}{\gfhei}{0}{0}{5}{2.5};
  \ovalTikz{\gfwid}{\gfhei}{0}{1}{9}{2.5};
  \ovalTikz{\gfwid}{\gfhei}{0}{0}{13}{2.5};
    
  \stateTikz{5}{4}{1}{1}{1}{1};
  \stateTikz{9}{4}{1}{1}{0}{1};
  \stateTikz{13}{4}{1}{1}{0}{0};

  \ovalTikz{\gfwid}{\gfhei}{0}{0}{3}{4.5};
  \ovalTikz{\gfwid}{\gfhei}{1}{0}{7}{4.5};
  \ovalTikz{\gfwid}{\gfhei}{0}{0}{11}{4.5};

  \stateTikz{3}{6}{1}{1}{1}{1};
  \stateTikz{7}{6}{1}{1}{1}{1};
  \stateTikz{11}{6}{1}{1}{0}{1};
  
  \ovalTikz{\gfwid}{\gfhei}{0}{0}{5}{6.5};
  \ovalTikz{\gfwid}{\gfhei}{0}{1}{9}{6.5};
  \ovalTikz{\gfwid}{\gfhei}{0}{0}{13}{6.5};


\draw[color=black] (1, 0.5) node {t};
\draw[color=black] (1, 2.5) node {t+1};
\draw[color=black] (1, 4.5) node {t+2};
\draw[color=black] (1, 6.5) node {t+3};

\draw[color=black] (3, -1) node {x-2};
\draw[color=black] (5, -1) node {x-1};
\draw[color=black] (7, -1) node {x};
\draw[color=black] (9, -1) node {x+1};
\draw[color=black] (11, -1) node {x+2};
\draw[color=black] (13, -1) node {x+3};

\end{tikzpicture}}
  \caption{\label{fig:gauge5}A spacetime diagram with $S_\psi = \mathbf{I}$ and $\widetilde{S}_J=\mathbf{I}$, initialized at $F=0$ with an initial condition on $A$ differing from that of subfigure (a). Observe that the $J$ field is as in $(a)$.}
\end{subfigure} \hfill\vspace{0.5cm} \centering
\begin{subfigure}[t]{0.45\textwidth}
  \centering
  \Large
  \resizebox{\textwidth}{!}{\begin{tikzpicture}


\draw[color=gray] (5.5,0) -- (12.5,7);
\draw[color=gray] (2.5,1) -- (8.5,7);
\draw[color=gray] (10,0.5) -- (14,4.5);
\draw[color=gray] (2,4.5) -- (4.5,7);

\draw[color=gray] (4,0.5) -- (2.5,2);
\draw[color=gray] (14,2.5) -- (9.5,7);
\draw[color=gray] (8,0.5) -- (2.5,6);
\draw[color=gray] (12,0.5) -- (5.5,7);

\node at (3, 1.5) [draw,scale=1,aspect=1,diamond,color=gray, fill=white]{}; 
\node at (7, 1.5) [draw,scale=1,aspect=1,diamond,color=gray, fill=white]{}; 
\node at (11, 1.5) [draw,scale=1,aspect=1,diamond,color=gray, fill=white]{}; 

\node at (5, 3.5) [draw,scale=1,aspect=1,diamond,color=gray, fill=white]{}; 
\node at (9, 3.5) [draw,scale=1,aspect=1,diamond,color=gray, fill=white]{}; 
\node at (13, 3.5) [draw,scale=1,aspect=1,diamond,color=gray, fill=white]{}; 

\node at (3, 5.5) [draw,scale=1,aspect=1,diamond,color=gray, fill=white]{}; 
\node at (7, 5.5) [draw,scale=1,aspect=1,diamond,color=gray, fill=white]{}; 
\node at (11, 5.5) [draw,scale=1,aspect=1,diamond,color=gray, fill=white]{};

  \stateTikz{5}{0}{1}{1}{0}{1};
  \stateTikz{9}{0}{1}{1}{0}{0};
  \stateTikz{13}{0}{1}{1}{1}{1};
  
  \ovalTikz{\gfwid}{\gfhei}{0}{0}{3}{0.5};
  \ovalTikz{\gfwid}{\gfhei}{0}{0}{7}{0.5};
  \ovalTikz{\gfwid}{\gfhei}{0}{1}{11}{0.5};

  \stateTikz{3}{2}{1}{1}{0}{0};
  \stateTikz{7}{2}{1}{1}{0}{1};
  \stateTikz{11}{2}{1}{1}{1}{1};
  
  \ovalTikz{\gfwid}{\gfhei}{0}{0}{5}{2.5};
  \ovalTikz{\gfwid}{\gfhei}{0}{1}{9}{2.5};
  \ovalTikz{\gfwid}{\gfhei}{0}{1}{13}{2.5};
    
  \stateTikz{5}{4}{1}{1}{0}{0};
  \stateTikz{9}{4}{1}{1}{1}{0};
  \stateTikz{13}{4}{1}{1}{0}{0};

  \ovalTikz{\gfwid}{\gfhei}{0}{0}{3}{4.5};
  \ovalTikz{\gfwid}{\gfhei}{0}{0}{7}{4.5};
  \ovalTikz{\gfwid}{\gfhei}{0}{1}{11}{4.5};

  \stateTikz{3}{6}{1}{1}{0}{0};
  \stateTikz{7}{6}{1}{1}{1}{0};
  \stateTikz{11}{6}{1}{1}{0}{1};
  
  \ovalTikz{\gfwid}{\gfhei}{0}{0}{5}{6.5};
  \ovalTikz{\gfwid}{\gfhei}{0}{1}{9}{6.5};
  \ovalTikz{\gfwid}{\gfhei}{0}{1}{13}{6.5};


\draw[color=black] (1, 0.5) node {t};
\draw[color=black] (1, 2.5) node {t+1};
\draw[color=black] (1, 4.5) node {t+2};
\draw[color=black] (1, 6.5) node {t+3};

\draw[color=black] (3, -1) node {x-2};
\draw[color=black] (5, -1) node {x-1};
\draw[color=black] (7, -1) node {x};
\draw[color=black] (9, -1) node {x+1};
\draw[color=black] (11, -1) node {x+2};
\draw[color=black] (13, -1) node {x+3};

\end{tikzpicture}}
  \caption{\label{fig:gauge3}A spacetime diagram with  $A(x,t+1) = A(x,t)$   and $\widetilde{S}_J=\mathbf{I}$, initialized with $F=1$ at position $x$ and $F=0$ everywhere else. Observe that the $J$ field differs from that of $(a)$.}
\end{subfigure}

\caption{\label{fig:gaugeinfluence}Examples of the influence of $F$, $\widetilde{S}_\bullet$ and $S_\bullet$}
\end{figure}

\paragraph{Gauge equivalence of two theories.\label{par:equivalency}} 
We need to keep in mind that by its very nature, such a $T$ cannot be unique---in the sense that for every candidate $T$ there will be several other physically equivalent local rules. This is because, as $T$ fully implements the local symmetry, it is inherently redundant, and thus equivalent to other theories up to this redundancy. 
\begin{definition}[Physically equivalent theories]\label{def:physeq} Two gauge-invariant theories $T$ and $T'$ are {\em physically equivalent theories} if and only if for any $T$--valid spacetime diagram $c$, there exists $\varphi$ such that $G_\phi c$ is a $T'$--valid spacetime diagram, and reciprocally.
\end{definition}
This definition ensures that given theory $T$ and some input configuration $c(.,t)$, we can always encode the input as $G_{\varphi(.,t)} c(.,t)$, and have it evolve under $T'$, so as to retrieve $G_{\varphi(.,t+1)} c(.,t+1)$, which is physically equivalent to  $c(.,t+1)$. We will now build candidate theories $T$ and $T'$ by following the standard steps of the gauging procedure. 

\section{Gauge field dynamics}\label{sec:gaugedynamics}
\paragraph{Dynamics of the invariant of the gauge field.} 
The dynamics $T$ we want to build takes $c(.,t)$ as input and outputs $c(.,t+1)$. However, we already have $R_\bullet$ which takes $c(.,t) = (\psi(.,t), A(.,t))$ and outputs $\psi(.,t+1)$. Therefore, all we need is a rule $S_\bullet$ that takes $c(.,t)$ and outputs $A(.,t+1)$. The standard procedure indeed proceeds by decomposing $T$ into $R_\bullet$ and $S_\bullet$. For such a $T$ to verify (\ref{eq:gaugeinvariance}), we just need $S_\bullet$ to verify (\ref{eq:inhomoggaugeinvariance}), for the same $G_\varphi$ and $Z$ that work for $R_\bullet$.\\

The procedure goes in two steps. The first step is to prescribe a dynamics $\widetilde{S}$ not over $A$, but over its invariant $\widetilde{A}$, which in our case amounts to a dynamics over $F$. It may even be a $J$--dependent dynamics $\widetilde{S}_\bullet$. Such a dynamics will be gauge-invariant by definition, since $F$ and $J$ are gauge-invariant. Thus the particular choice of $\widetilde{S}_\bullet$ is only dictated by the phenomena that we wish to model. A simple choice, for instance, is to take $\widetilde{S}$ to be the identity. Then, if initially we had $F=0$ initially, this will remain the case. Beware that this does not mean that the behaviour of the underlying $A$ will be trivial. In fact it will remain largely undetermined, as $F=0$ just means $\Delta_r A_l(x,t)=\Delta_l A_r(x,t)$. But at least this constraint over the dynamics of $A$ is gauge-invariant. Fig-\ref{fig:gauge1}, \ref{fig:gauge2}, \ref{fig:gauge5} give examples of different $A$ that have $F=0$---illustrating how many different gauge fields can arise from the same prescription for the invariant. Thus, this first step does not suffice to prescribe $S_\bullet$. Hence the need for a second step called gauge-fixing.

\paragraph{Gauge-fixing : completing the dynamics.} 
Gauge-fixing means choosing an actual $S_\bullet$ 
which induces the $\widetilde{S}_\bullet$ that we had settled for in the first step. In our case, we need to fix an $S_\bullet$ such that for all $(\psi,A)$, if $A$ is an $S_\psi$--valid spacetime diagram, then $F$ is a $\widetilde{S}_J$--valid spacetime diagram, with $F$ and $J$ computed from $A$ and $\psi$.\\
It is crucial to understand that this time the choice of a particular $S_\bullet$ is not dictated by the physics, but by mere convenience. This assertion relies on the following proposition.
\begin{proposition}[Gauge-fixing soundness] Let $R_\bullet$ be an $A$--dependent inhomogeneous theory upon $\psi$, with respect to a given $G_\varphi$ and $Z$. Let $S_\bullet$ and $S'_\bullet$ be two $\psi$--dependent inhomogeneous gauge-invariant theories upon $A$, with respect to the same $G_\varphi$ and $Z$. If the two theories induce the same $\widetilde{S}_\bullet$, then: 
\begin{enumerate}
	\item For any $A$ an $S_\psi$--valid spacetime diagram, there exists $\varphi$ such that $A'$ is an $S'_{\psi'}$--valid spacetime diagram, with $(A',\psi')=G_\varphi(A,\psi)$.
    \item The theories $T=R_\bullet\wedge S_\bullet$ and $T'=R_\bullet\wedge S'_\bullet$ are physically equivalent theories.
\end{enumerate}
\end{proposition}
\begin{proof}
\begin{enumerate}
\item $S_\bullet$ and $S'_\bullet$ induce the same $\widetilde{S}_\bullet$. Giving $A$ an $S_\bullet$--valid field also gives $F$ an $\widetilde{S}_\bullet$--valid field. We can then build $A'$ an $S'_\bullet$-valid field, inducing the same $F$. This is done using an initial condition for $A'$ which gives $F(.,0)$. When evolving with $S'_\bullet$, $F$ will evolve with $\widetilde{S}'_A$. Using remark-\ref{prop:gaugeinv} which says that given $A$ and $A'$ inducing the same $F$, there exists $\varphi$ such that $A'=G_\varphi A$, we prove Fact 1 by applying such a $G_\varphi$ to $c=(\psi,A)$. Thus we have built a $\varphi$ such that $A'$ is an $S'_{\psi'}$--valid spacetime diagram with $(\psi',A') = G_\varphi (\psi, A)$.

\item We can now prove Fact 2. Given $c=(\psi,A)$ a $T$--valid spacetime diagram, consider $A$ on its own. $A$ is an $S_\psi$--valid spacetime diagram. But since $S_\bullet$ and $S'_\bullet$ both implement $\widetilde{S}_\bullet$, there exists $\varphi$ such that $A'=G_\varphi A$ is an $S'_{G_\varphi\psi}$--valid spacetime diagram (Fact 1). Apply this $G_\varphi$ to the whole of $c=(\psi,A)$. This yields some $c'=(\psi',A')$. Is $c'$ a $T'$--valid spacetime diagram? Yes, because: $A'$ is an $S'_{\psi'}$--valid spacetime diagram by construction and since $\psi$ is an $R_A$--valid spacetime diagram,  $\psi'$ is an $R_{A'}$--valid spacetime diagram due to $R_\bullet$ gauge-invariance. Hence $T$ and $T'$ are physically equivalent theories.
\end{enumerate}
\hfill$\square$
\end{proof}

Hence, different dynamics $S_\bullet$ can be used to describe the same physics, and choosing between them is a matter of convenience. To illustrate this point, we refer again to Fig-\ref{fig:gaugeinfluence}. In this figure, we can see that sub-figures \ref{fig:gauge1}, \ref{fig:gauge2}, \ref{fig:gauge5} are physically equivalent, the same gauge-field invariant $F=0$ is used. But sub-figure \ref{fig:gauge2} shows that we do have a degree of freedom on $S_\psi$. 

\paragraph{Gauge-constraining : removing redundancies.}
Now that we fully described our gauge-invariant theory $T$, we find ourselves confronted with its inherent redundancies---the ones arising precisely from the gauge-symmetry we just managed to implement. Indeed, any two $G_\varphi$--related initial configurations, lead to physically equivalent solutions---as shown in Fig-(\ref{fig:gauge1} and \ref{fig:gauge5}). At this stage, and only now that the symmetry has been implemented, we may wish to remove its induced redundancies by suitably restricting the space of configurations. This is usually done by imposing some local constraints directly on the field $c(x,t)$, referred to as gauge-constraining.  However we must keep in mind that constraining $c(x,t)$ could potentially restrict the set of physical solutions available. One must therefore check that a gauge-invariant theory $T$, and its gauge-constrained version $T''$, remain physically equivalent theories in the sense of Def. \ref{def:physeq}.

\section{Conclusion}\label{sec:conclusion}

\noindent {\em Summary.} The paper followed a discrete counterpart to the gauging procedure, which aims to enforce a local symmetry that was judged missing in some physical theory. Here, theories were captured as Cellular Automata (CA), and local symmetries as local transformations $G_\varphi$ of the spacetime diagrams $c$ of these CA. Gauge-invariance was formulated as a concrete condition \eqref{eq:gaugeinvariance}, which directly translates into a local constraint upon the local rule of the theory. It was shown how, starting from a homogeneous non-gauge-invariant theory $R$ over configurations $\psi(.,t)$, one gets to an $A(x,t)$--dependent inhomogeneous gauge-invariant theory $R_\bullet$, and completes this with a $\psi(x,t)$--dependent gauge-invariant theory $S_\bullet$ over configurations $A(.,t)$, in order to finally obtain a homogeneous gauge-invariant theory $T=R_\bullet\cup S_\bullet$ over configurations $c(.,t)=(\psi(.,t),A(.,t))$. The acquired gauge-symmetry then leads to equivalent theories $T'$---equivalent up to the symmetry. A way to go from a $T$ to some equivalent $T'$ is to replace $S_\bullet$ by some $S'_\bullet$ whose spacetime diagrams are $G_\varphi$--related---this is called gauge-fixing. Theory equivalence and gauge-fixing were formalized, the fact that the latter respects the former was proven. Moreover, one can sometimes find an equivalent theory on a reduced configuration space $\widetilde{c}(.,t)$, which can be understood as a canonical representant of $c(.,t)$ under the symmetry---this is called gauge-constraining. 

\noindent {\em Motivations.} These were twofold: (i) Porting the gauge theoretical tools and concepts to Computer Science, as methods for constructing nature-inspired CA; providing more accurate schemes for numerical analysis; providing quantum simulation algorithms; making spatially distributed (quantum) computation immune to local errors. (ii) Clarifying the gauge theoretical concepts through the simplicity and rigor brought by Discrete Mathematics; providing the most direct route to its core, i.e. without reference to quantum mechanics and least action principle.

\noindent {\em Related works.} A number of discrete counterparts to physics symmetries have been reformulated in terms of CA, including reversibility, Lorentz-covariance\cite{arrighi2014discrete}, conservations laws and invariants\cite{formenti2011hierarchy}, but no gauge symmetry. To our knowledge the closest work is the colour-blind CA construction\cite{salo2013} which implements a global colour symmetry without porting it to the local scale. However gauge symmetries have been implemented in the one-particle sector of Quantum CA, a.k.a for Quantum Walks. Indeed, one of the authors had followed a similar procedure in order to introduce the electromagnetic gauge field \cite{di2014quantum}, and that of the weak and strong interactions \cite{arnault2016quantum,di2016quantum}. This again was done in the very fabric of the Quantum Walk and the associated symmetry was therefore an intrinsic property of the Quantum Walk. But the gauge field would remain continuous, and seen as an external field.\\
There are, of course, numerous other approaches to space-discretized gauge theories, the main ones being Lattice Gauge Theory\cite{Willson1987} and the Quantum Link Model\cite{chandrasekharan1997quantum}, which were phrased in terms of Quantum Computation--friendly terms through Tensor Networks\cite{rico2014tensor} and can be linked in a unified framework\cite{silvi2014}. Discretized gauge-theories have also arisen from Ising models\cite{silvi2014,wegner1971}.
All of these approaches, however, begin with a well-known continuous gauge theory which is then space-discretized---time is usually kept continuous. An interesting attempt to quantum discretize gauge theories in discrete time, on a general simplicial complex can be found in \cite{kornyak2009discrete}.
  
\noindent {\em Perspectives.}  We believe that the hereby developed methodology is ready to be applied to Quantum CA (QCA) \cite{arrighi2008one}, so as to obtain discretized free and interacting Quantum Field Theories \cite{itzykson2006quantum}. Such discretized theories are of interest in Physics especially in non-perturbative theories \cite{strocchi2013introduction}, but they also represent practical assets as quantum simulation algorithms, i.e. numerical schemes that run on Quantum Computers to efficiently simulate interacting fundamental particles---a task which is beyond the capabilities of classical computers. This is ongoing work. 

\section{Acknowledgements}
The authors would like to thank C\'edric B\'eny, Thomas Krajewski, Terry Farrelly and Pablo Arnault, for very instructive conversations about gauge theories. This work was partially supported by the CNRS PEPS JCJC \textsc{GQNet} and the CNRS PEPS D{\'e}fi InFinitTI ``Lattice Quantum Simulation Theory'' \textsc{LaQuST}. 
\bibliographystyle{llncs/splncs}
\bibliography{biblio}
\end{document}